\documentclass[a4paper,12pt]{article}
\usepackage{amsmath,amsthm}
\theoremstyle{plain}
\newtheorem*{theorem}{Theorem}
\newtheorem*{lemma}{Lemma}

\numberwithin{equation}{section}
\author{Nikolay K. Vitanov$^{1,2,}$\footnote{corresponding author, e-mail: vitanov{@}imbm.bas.bg}, Zlatinka I. Dimitrova$^3$, Tsvetelina I. Ivanova$^4$}
\date{$^1$Institute of Mechanics, Bulgarian Academy of Sciences, Acad. G. Bonchev Str.,
Bl. 4, 1113 Sofia, Bulgaria\\ 
$^2$Max-Planck Institute for the Physics of Complex Systems, N{\"o}thnitzer Str. 38,
01187 Dresden, Germany\\ 
$^3$ "G. Nadjakov" Institute for Solid State Physics, Bulgarian Academy of Sciences,
Blvd. Tzarigradsko Chaussee 72, 1784, Sofia, Bulgaria\\ 
$^4$ Faculty of Physics, "St. Kliment Ohridski" University of Sofia, Blvd. J. Bourchier 5,
1164 Sofia, Bulgaria}
\title{Solitary wave solutions of nonlinear partial differential equations based on the simplest equation for the function $1/\cosh^n$}
\begin{document}
\maketitle
\begin{abstract} 
The method of simplest equation is applied for obtaining 
exact solitary traveling-wave solutions of nonlinear partial 
differential equations that contain monomials of odd and even 
grade with respect to participating derivatives. The used simplest equation is
$f_\xi^2 = n^2(f^2 -f^{(2n+2)/n})$. The developed methodology is illustrated
on two examples of classes of nonlinear partial differential equations that
contain: (i) only monomials of odd grade with respect to participating
derivatives; (ii) only monomials of even grade with respect to participating
derivatives. The obtained solitary wave solution for the case (i) contains as
particular cases the solitary wave solutions of Korteweg-deVries equation and
of a version of the modified Korteweg-deVries equation.
\end{abstract}
%
%
\section{Introduction}

The methods of nonlinear dynamics, nonlinear time series analysis and 
theory of differential equations are highly actual because of their numerous applications in data analysis,  theory of chaos, social dynamics, etc. \cite{p1} - \cite{p23}.
Traveling wave solutions of the nonlinear partial differential equations are
studied much in the last decades \cite{scott}-\cite{kudr90} as such waves exist 
in many natural systems \cite{infeld}-\cite{ablowitz1}. In addition
effective methods exist for obtaining exact traveling-wave solutions, e.g., 
the method of inverse scattering transform or the method of Hirota \cite{gardner} - \cite{hirota}. Many other approaches for obtaining exact special
solutions of nonlinear PDEs have been developed in the recent years. Let us note here
only the method of simplest equation and its version called modified method of simplest
equation \cite{kudr05x} - \cite{v15} as these methods are closely connected to 
the discussion below. 
\par
The method of simplest equation is based on  a procedure analogous to the first 
step of the test for the Painleve property \cite{kudr05x}, \cite{kudr08}, 
\cite{kudr05}-\cite{k10}. In the version of the method called  modified method of 
the simplest equation \cite{vdk10} - \cite{v15} this procedure is substituted by  the concept for the balance equation. This version of the method of simplest equation has
been successfully applied for obtaining exact traveling wave solutions of  numerous
nonlinear PDEs such as versions of generalized Kuramoto - Sivashinsky equation, reaction - diffusion  equation, reaction - telegraph equation\cite{vdk10}, \cite{vd10}
generalized Swift - Hohenberg equation and generalized Rayleigh equation \cite{vit11}, 
generalized Fisher equation, generalized Huxley equation \cite{vit09x}, generalized Degasperis - Procesi equation and b-equation \cite{vit11a}, extended Korteweg-de Vries equation \cite{v13}, etc.  \cite{v13a}, \cite{v14}.
\par
A short summary of the method of simplest equation we shall use below is as follows. First of all by means of an appropriate ansatz (for an example the traveling-wave ansatz) the solved  nonlinear partial differential equation is reduced to a nonlinear
ordinary differential equation
\begin{equation}\label{i1}
P \left( u, u_{\xi},u_{\xi \xi},\dots \right) = 0
\end{equation}
Then the solution $u(\xi)$ is searched as some function of another function $f(\xi)$.
Often this function is a finite-series solution 
\begin{equation}\label{i2}
u(\xi) = \sum_{\mu=-\nu}^{\nu_1} p_{\mu} [f (\xi)]^{\mu},
\end{equation}
where $p_\mu$ are coefficients and $f(\xi)$ is solution of
simpler ordinary differential equation called simplest equation. 
Eq.(\ref{i2}) is substituted in Eq.(\ref{i1}) and let the result
of this substitution be a polynomial of $f(\xi)$. Then Eq. (\ref{i2}) is a candidate
for solution of Eq.(\ref{i1}) if all coefficients of the obtained polynomial of $f(\xi)$ 
are equal to $0$. This condition leads to a system of nonlinear algebraic 
equations for the coefficients of the solved nonlinear PDE and for the
coefficients of the solution. Any nontrivial solution of this algebraic system
leads to a solution the studied  nonlinear partial differential equation.
\par
Below we shall consider traveling-wave solutions 
\begin{math}
u(x,t) = u( \xi ) = u(\alpha x + \beta t), 
\end{math} constructed on the basis of the simplest equation 
\begin{equation}\label{se}
f_{\xi}^2 = n^2 \left( f^2 - f^{\frac{2n+2}{n}} \right),
\end{equation}
where $n$ is \emph{arbitrary positive real number}.
The solution of this equation is $f( \xi ) = \frac{1}{\cosh^{n}(x)}$.
The text is organized as follows. In Sect. 2 we  describe the 
main result which is formulated as a theorem. In section 3 we discuss several
examples for solitary wave solutions obtained by the developed methodology. 
Several concluding remarks  are summarized in Sect.4.
\section{Main result}
We shall use the concept of grade of a monomial with respect to participating derivatives.
Let us consider polynomials that are linear combination of monomials and the 
monomials contain product of terms  consisting of product of powers of derivatives of
some function $u(\xi)$ of different orders. This product of terms can be 
multiplied by a polynomial of $u$. Let a term from a monomial contains $k$th 
power of a derivative of
$u$ of $l$th order. We shall call the product $kl$ grade of the term with respect to 
participating derivatives. The sum of these grades of all terms of the monomial will be called grade of the monomial with respect to participating derivatives. In the general case the 
polynomial will contain monomials of odd and even grades with respect to participating derivatives. There are two particular cases: (i) the polynomial contains monomials 
that are only of odd grades with respect to participating derivatives; (ii) the polynomial contains monomials that are only of even grades with respect to participating derivatives.
Below we formulate a theorem about solitary wave solutions for the class of nonlinear PDEs 
that contain monomials of derivatives which order with respect to 
participating derivatives is even and monomials of derivatives which 
order with respect to participating derivatives is odd.
\par
First of all we shall use the method of induction to prove an useful lemma and
then we shall prove a theorem that will be our main result.
\begin{lemma}
Let $f(\xi)$ be a function that is solution of Eq.(\ref{se}). Let us consider
the function $F(\xi)=f^{(in+2j)/n}$ where $n$ is a positive real number and 
$i$ and $j$ are non-negative integer numbers. Then the odd derivatives of $F$ 
contain $f_\xi$ multiplied by expression(s) of the kind $f^{(kn+2j)/n}$ and 
the even derivatives of $F$ contain expressions of the kind $f^{(kn+2j)/n}$,
where $k$ is some integer number.
\end{lemma}
\begin{proof}
Let us consider first the case $i=0$. Then $F(\xi)=f^{2j/n}$ The first and the
second derivatives of $F(\xi)$ are
\begin{equation}\label{fd1}
F_\xi = \frac{2n}{j} f^{(2j-n)/n} f_\xi
\end{equation}
\begin{equation}\label{sd1}
F_{\xi \xi} = \frac{2j}{n} \frac{2j-n}{n}f^{(2j-2n)/n}f_\xi^2
+ \frac{2j}{n} f^{(2j-n)/n}f_{\xi \xi}
\end{equation}
$f_\xi^2$ can be substituted from Eq.(\ref{se}) and from the same equation one
obtains
\begin{equation}\label{se2}
f_{\xi \xi} = n^2 \left(f- \frac{n+1}{n} f^{(n+2)/n} \right)
\end{equation}
As one can see $F_\xi$ contains $f_\xi$ multiplied by a term of the kind
$f^{(kn+2j)/n}$, and the substitution of Eqs. (\ref{se}) and (\ref{se2}) in 
Eq.(\ref{sd1}) leads to the conclusion that the even derivative $F_{\xi\xi}$ contains 
only expressions of the kind $f^{(kn+2j)/n}$ (namely $f^{2j/n}$ and  
$f^{2(j+1)/n}$) as stated in the lemma.
The calculations can be made further. The result is that the third derivative
$F_{\xi\xi\xi}$ (which is an odd derivative) contains $f_\xi$ multiplied by a
sum  of expressions of the kind
$f^{(kn+2j)/n}$. The fourth derivative $F_{\xi\xi\xi \xi}$ (which is an even
derivative) contains only sum of expressions of the kind $f^{(kn+2j)/n}$, etc.
\par
Let us now consider the derivative of $F$ of order $2q$ where $q$ is a natural
number. According to the lemma we assume that this derivative (denoted as
$F^{(2q)}_\xi)$) contains only sum of expressions of the kind $f^{(kn+2j)/n}$.
Then the obtaining  the derivative $F^{(2q+1)}_\xi$ includes many operations similar
to the operation of obtaining  $F_\xi$ from Eq.(\ref{fd1}). The result for
$F^{(2q+1)}_\xi$ will be an expression consisting of  $f_\xi$ multiplied by 
terms of the kind $f^{(kn+2j)/n}$. The next even derivative $F^{(2q+2)}_\xi$ will be obtained
in a way similar to obtaining $F_{\xi \xi}$ from Eq.(\ref{sd1}) and the relationship for 
$F^{(2q+2)}_\xi$ will contain only expressions of the kind  $f^{(kn+2j)/n}$.
This concludes the proof of the lemma for the case $i=0$.
\par
Let now $i>0$.
The first derivative of $F(\xi)$ is
\begin{equation}\label{lx1}
F_\xi = \frac{in+2j}{n}f^{[(i-1)n+2j]/n} f_\xi
\end{equation}
This is an odd derivative and it contains $f_\xi$ and expression of the kind
$f^{(in+2j)/n}$ as stated in the lemma. The second derivative of $F(\xi)$ is
\begin{eqnarray}\label{lx2}
F_{\xi \xi} = \frac{in+2j}{n} \frac{(i-1)n+2j}{n} f^{[(i-2)n+2j]/n}  f_\xi^2
+ \frac{in+2j}{n}f^{[(i-1)n+2j]/n} f_{\xi \xi}
\end{eqnarray}
$f_\xi^2$ can be substituted from Eq.(\ref{se}) and $f_{\xi \xi}$ can be
substituted from Eq.(\ref{se2}). 
The substitution of Eqs(\ref{se}) and (\ref{se2}) in Eq.(\ref{lx2}) leads to the
conclusion that the even derivative $F_{\xi\xi}$ contains 
only expressions of the kind $f^{(in+2j)/n}$
(namely $f^{(in+2j)/n}$ and  $f^{[in+2(j+1)]/n}$) as stated in the lemma.
The calculations can be made further. The result is that the third derivative
$F_{\xi\xi\xi}$ (which is an odd derivative) contains $f_\xi$ multiplied by a
sum  of expressions of the kind
$f^{(in+2j)/n}$. The fourth derivative $F_{\xi\xi\xi \xi}$ (which is an even
derivative) contains only sum of expressions of the kind $f^{(in+2j)/n}$, etc.
\par
Let us now consider the derivative of $F$ of order $2q$ where $q$ is a natural
number. According to the lemma we assume that this derivative (denoted as
$F^{(2q)}_\xi)$) contains only sum of expressions of the kind $f^{(in+2j)/n}$.
Then the obtaining the derivative $F^{(2q+1)}_\xi$ includes many operations similar
to the operation of obtaining of $F_\xi$ from Eq.(\ref{fd1}). The result for
$F^{(2q+1)}_\xi$ will be an expression consisting of  $f_\xi$ multiplied by terms of
the $f^{(in+2j)/n}$. Then the next even derivative $F^{(2q+2)}_\xi$ be obtained
in a way similar to obtaining $F_{\xi \xi}$ from Eq.(\ref{sd1}) and the relationship for 
$F^{(2q+2)}_\xi$ will contain only expressions of the kind  $f^{(in+2j)/n}$.
This concludes the proof of the lemma. 
\end{proof}
Now we are ready to formulate and prove our main result.
\begin{theorem}
Let {$\cal{P}$} be a polynomial of the function $u(x,t)$ and its derivatives. $u(x, t)$ belongs to the differentiability class $C^k$, where $k$ is
the highest order of derivative participating in {$\cal{P}$}. {$\cal{P}$} can contain some or all of the following parts: (A) polynomial of $u$; (B)
monomials that contain derivatives of $u$ with respect to $x$ and/or products of such derivatives. Each such monomial can be
multiplied by a polynomial of $u$; (C) monomials that contain derivatives of $u$ with respect to $t$ and/or products of such derivatives.
Each such monomial can be multiplied by a polynomial of $u$; (D) monomials that contain mixed derivatives of $u$ with respect to $x$ and $t$ and/or products of such derivatives. Each such monomial can be multiplied by a polynomial of $u$; (E) monomials that contain products of derivatives of $u$ with respect to $x$ and derivatives of $u$ with respect to $t$. Each such monomial can be multiplied by a polynomial of $u$; (F) monomials that contain products of derivatives of $u$ with respect to $x$ and mixed derivatives of $u$ with respect to $x$ and $t$. Each such monomial can be multiplied by a polynomial of $u$; (G) monomials that contain products of derivatives of $u$  with respect to $t$ and mixed derivatives of $u$  with respect to $x$  and $t$ . Each such monomial can be multiplied by a polynomial of $u$ ; (H) monomials that contain products of derivatives of $u$ with respect to $x$ , derivatives of $u$  with respect to $t$  and mixed derivatives of $u$  with respect to $x$  and $t$ . Each such monomial can be multiplied by a polynomial of $u$.
\par
Let us consider the nonlinear partial differential equation:
\begin{equation}\label{nde}
{\cal{P}}=0
\end{equation}
We search for solutions of this equation of the kind 
$u( \xi) = \gamma f( \xi ), \xi = \alpha x + \beta t$. $\gamma $ is a parameter and $f( \xi )$ is a solution of the 
simplest equation $f_{\xi}^2 = n^2 f^2 - n^2 f^{\frac{2n+2}{n}}$ where $n$ is a
real positive number.
The substitution of this solution in (\ref{nde}) leads to a relationship R of the kind
\begin{equation}\label{r}
R = \sum_{i=0}^{N_1} \sum_{j=0}^{M_1} C_{ij} f( \xi )^{(in+2j)/n} + f_{\xi} 
\sum_{k=-N_2^*}^{N_2} \sum_{l=0}^{M_2} D_{kl} f( \xi )^{(kn+2l)/n},
\end{equation}
 where $N_1,N_2,N_2^*,M_1$ and $M_2$ are natural numbers depending on the form of the polynomial P.
The coefficients $C_{ij}$ and $D_{kl}$ depend on the parameters of 
Eq.(\ref{nde}) and on $\alpha, \beta, \gamma$. 
The sum $\sum \limits_{i=0}^{N_1} \sum \limits_{j=0}^{M_1} C_{ij} f( \xi )^{(in+2j)/n}$ consists
of terms of the kind $C_p^* f(\xi)^{\alpha_p}$ where $\alpha_p$ is some number and
$p=1,\dots$. The sum $\sum \limits_{k=0}^{N_2} \sum \limits_{l=0}^{M_2} D_{kl} f( \xi
)^{(kn+2l)/n}$ consists of terms of the kind $D_q^* f(\xi)^{\beta_q}$ where
$\beta_q$ is some number and $q=1,\dots$. The coefficients $C_p^*$ and $D_q^*$ depend on the parameters of 
Eq.(\ref{nde}) and on $\alpha, \beta, \gamma$
Then any nontrivial solution of the algebraic system 
\begin{eqnarray}\label{rel}
C_p^* = 0;  \ \  D_q^* = 0, 
\end{eqnarray}
leads to a solitary wave solution of the nonlinear PDE (\ref{nde}).
\end{theorem}
\begin{proof}
Let $f( \xi )$ be a solution of the nonlinear ODE  $f_{\xi}^2 = n^2 f^2 - n^2
f^{\frac{2n+2}{n}}$. According to the lemma above the higher derivatives of 
$f(\xi)$ contain terms of the kind  $f^{(in+2j)/n}$ or $f_\xi$ multiplied by
terms of the kind $f^{(kn+2j)/n}$. The substitution of these derivatives in 
the solved nonlinear PDE ${\cal{P}}=0$ will reduce the solved equation
to a relationship of the kind (\ref{r}). In order to obtain a solution of
Eq.(\ref{r}) we have to solve the system of equations (\ref{rel}), i.e., a
system of nonlinear algebraic equations for $\alpha$, $\beta$, $\gamma$ and the
parameters participating in $\cal{P}$. Any nontrivial solution of the last
system of nonlinear algebraic equation leads to a solution of Eq.(\ref{nde}) of
the kind $u=\gamma f(\xi)$ where $f(\xi)$ is solution of the simplest equation 
(\ref{se}) (note that $n$ is an arbitrary real positive number).
\end{proof}
Let us note two particular cases connected to the values $n=1$ and $n=2$. For
the case $n=1$ the simplest equation (\ref{se}) becomes 
\begin{equation}\label{se3}
f_{\xi}^2 =  f^2 - f^4 
\end{equation}
and its solution is $f( \xi ) = \frac{1}{\cosh(x)}$. In this case $R$ from
Eq.(\ref{r}) can be reduced to 
\begin{equation}\label{r1}
R = \sum_{i=0}^{N}  C_{i} f( \xi )^i + f_{\xi} \sum_{j=0}^{M}  D_{j} f( \xi )^j,
\end{equation}
and the system of nonlinear algebraic equations becomes $C_i=0,D_j=0$; 
$i=0,\dots$, $j=0,\dots$.  The other particular case is $n=2$. Here the simplest equation becomes
\begin{equation}\label{se4}
f_{\xi}^2 =  4(f^2 - f^3) 
\end{equation}
and the solution is $f( \xi ) = \frac{1}{\cosh^2(x)}$.
 For this particular case $R$ again can be reduced to the relationship of the
kind (\ref{r1}) and the described methodology leads to 
the solitary wave solutions of many famous water-waves equations such as the
Korteweg-deVries equation, Boussinesq equation, Degasperis-Processi equation,
etc. \cite{v14}.
\section{Examples}
We shall consider examples of nonlinear partial differential equations that contain
monomials only of odd and even grade with respect to participating derivatives.
\subsection{Case of a nonlinear partial differential equation that contains monomials
only of odd grades with respect to participating derivatives}
Let us consider as an example  the equation
\begin{equation}\label{e1}
a F^p \frac{\partial ^\mu F}{\partial x^ \mu} + bF^q \frac{\partial^\nu 
F}{\partial t^\nu}  + cF^r \frac{\partial 
F}{\partial x}= 0
\end{equation}
where $a$,$b$,$p$,$q$, $\mu$, and $\nu$ are parameters. We search for a solution
of the kind $F=\gamma f$ where $\gamma$ is a parameter and $f(x,t) = f(\xi)$;
$\xi = \alpha x + \beta t$ is solution of the simplest equation (\ref{se}).
The substitution of $F$ in Eq.(\ref{e1}) leads to the following
equation for $f(\xi)$
\begin{equation}\label{ex1}
\alpha^\mu \gamma^{p+1} a f^p \frac{d^\mu f}{d \xi^\mu} + \beta^\nu \gamma^{q+1}
b f^q \frac{d^\nu f}{d \xi^\nu} + \alpha c \gamma^{r+1} f^r \frac{df}{d \xi} =0
\end{equation}
The most simple case of nonlinear equation that contains only odd derivatives is
$\mu=1$, $\nu=3$ or $\mu=3$, $\nu=1$. 
\subsubsection{Case $\mu=1$, $\nu=3$}
In this case the substitution of Eq.(\ref{se}) in Eq.(\ref{ex1}) 
leads to the relationship
\begin{equation}\label{ae1}
\alpha \gamma^{p+1} a f^p + \beta^3 \gamma^{q+1}n^2 b f^q  - \beta^3
\gamma^{q+1}(n+1)(n+2) b f^{q+2/n} + \alpha c \gamma^{r+1} f^r =0
\end{equation}
The relationship (\ref{ae1}) contains several powers of the
function $f$. In order to obtain the system of the nonlinear algebraic
relationships we have to perform the balance procedure from the modified method
of simplest equation. As a result we have two possibilities: (i) $p = q +2/n$, $r=q$; (ii)
$r=q+2/n$, $p=q$. For the case (i) the equation (\ref{e1}) becomes 
\begin{equation}\label{e1x1}
\left(aF^{2/n} + c  \right) \frac{\partial F}{\partial x} + b \frac{\partial^3 F}{\partial t^3} =0
\end{equation}
We obtain a system of two nonlinear algebraic
equations. The solution of this system is
\begin{equation}\label{sol1}
\gamma = \left[ - \frac{(n+1)(n+2)}{n^2} \frac{c}{a}\right]^{n/2}; \ \ 
\alpha = - \frac{\beta^3 b n^2}{c}
\end{equation}
and the corresponding solution of the equation (\ref{e1x1}) is
\begin{equation}\label{seq1}
F_{(n)}(x,t) = \frac{\left[- \frac{(n+1)(n+2)}{n^2} \frac{c}{a} \right]^{n/2}}{
\cosh^n \left[ - \frac{n^2 \beta^3 b}{c}x + \beta t \right]}
\end{equation}
Now for $n=1$  and for $n=2$ we obtain the solitary wave solutions
\begin{equation}\label{seq2}
F_{(1)}(x,t) = \frac{\left[- 6 \frac{c}{a} \right]^{1/2}}{
\cosh \left[ - \frac{ \beta^3 b}{c}x + \beta t \right]}; \ \ 
F_{(2)}(x,t) = \frac{\left[- 3 \frac{c}{a} \right]}{
\cosh^2 \left[ - \frac{4 \beta^3 b}{c}x + \beta t \right]},
\end{equation}
for $n=10$  and for $n=1/4$ we obtain the solitary wave solutions
\begin{equation}\label{seq4}
F_{(10)}(x,t) = \frac{\left[- \frac{132}{100} \frac{c}{a} \right]^{5}}{
\cosh^{10} \left[ - \frac{100 \beta^3 b}{c}x + \beta t \right]}; \ \ 
F_{(1/4)}(x,t) = \frac{\left[- 45 \frac{c}{a} \right]^{1/8}}{
\cosh^{1/4} \left[ - \frac{ \beta^3 b}{16 c}x + \beta t \right]}
\end{equation}
Note that $n$ can  be arbitrary positive real number. For an example for
$n=2.22$ the solitary wave solution of Eq.(\ref{e1x1}) is
\begin{equation}\label{seq6}
F_{(2.22)}(x,t) = \frac{\left[- 2.7571626 \frac{c}{a} \right]^{1.11}}{
\cosh^{2.22} \left[ - \frac{4.9284 \beta^3 b}{c}x + \beta t \right]}
\end{equation}
\par
Let us now consider the case (ii). Eq.(\ref{e1}) becomes
\begin{equation}\label{e1x2}
\left(a + c F^{n/2} \right) \frac{\partial F}{\partial x} + b \frac{\partial^3 F}{\partial t^3} =0
\end{equation}
which is equation of the same kind as Eq.(\ref{e1x1}).
In this case the balance procedure leads 
again to a system of two nonlinear algebraic equations. The solutions of this
system is
\begin{equation}\label{sol2}
\gamma = \left[ - \frac{(n+1)(n+2)}{n^2} \frac{a}{c}\right]^{n/2}; \ \ 
\alpha = - \frac{\beta^3 b n^2}{a}
\end{equation}
and in this case the solution becomes
\begin{equation}\label{seq7}
F_{(n)}(x,t) = \frac{\left[- \frac{(n+1)(n+2)}{n^2} \frac{a}{c} \right]^{n/2}}{
\cosh^n \left[ - \frac{n^2 \beta^3 b}{a}x + \beta t \right]}
\end{equation}
\subsubsection{Case $\mu=3$, $\nu=1$}
Let us now consider the case $\mu=3$, $\nu=1$. The substitution of the 
form of $F$ and the derivatives of $f$ into Eq.(\ref{ex1}) leads  to the 
relationship
\begin{equation}\label{rel1}
\alpha^3 a n^2 \gamma^{p+1}f^p - \alpha^3 a (n+1)(n+2) \gamma^{p+1} f^{p+2/n} +
\beta \gamma^{q+1} b f^q + \alpha c \gamma^{r+1} f^{r}=0
\end{equation}
There are two
possibilities: (i) $q=p+2/n$,$r=p$ ; (ii) $r=p+2/n$, $q=p$. 
For the case (i) Eq.(\ref{e1}) becomes
\begin{equation}\label{e1x3}
a  \frac{\partial ^3 F}{\partial x^3} + bF^{2/n} \frac{\partial 
F}{\partial t}  + c \frac{\partial 
F}{\partial x}= 0
\end{equation}
The obtained system of nonlinear algebraic equations has
a solution
\begin{equation}\label{c2s1}
\alpha = \pm \left( - \frac{c}{an^2}\right)^{1/2}; \ \beta = \mp \left( - 
\frac{c}{an^2} \right)^{3/2} \frac{a}{b} (n+1)(n+2) \gamma^{-2/n}
\end{equation}
and the solitary wave solution of Eq.(\ref{e1x1}) is
\begin{equation}\label{swx3}
F(x,t) = \frac{\gamma}{\cosh^n \Bigg[\pm \left( - \frac{c}{an^2}\right)^{1/2} x 
  \mp \left( - 
\frac{c}{an^2} \right)^{3/2} \frac{a}{b} (n+1)(n+2) \gamma^{-2/n} t\Bigg]}
\end{equation}
Note that $n$ is arbitrary finite positive real number.
\par 
For the case (ii) Eq.(\ref{e1}) becomes
\begin{equation}\label{e1x4}
a  \frac{\partial ^3 F}{\partial x^3} + b \frac{\partial 
F}{\partial t}  + cF^{n/2} \frac{\partial 
F}{\partial x}= 0
\end{equation}
and the obtained system of nonlinear algebraic equations has a
solution
\begin{equation}\label{c2s2}
\alpha = \pm \left[ \frac{c\gamma^{2/n}}{a(n+1)(n+2)}\right]^{1/2}; \ 
\beta = \mp \frac{an^2}{b} \left[ \frac{c \gamma^{2/n}}{a(n+1)(n+2)}\right]^{3/2}
\end{equation}
The solitary wave solution of Eq.(\ref{e1x4}) is
\begin{equation}\label{swx3}
F(x,t) = \frac{\gamma}{\cosh^n \Bigg[ 
\pm \left[ \frac{c\gamma^{2/n}}{a(n+1)(n+2)}\right]^{1/2} x  \mp
\frac{an^2}{b} \left[ \frac{c \gamma^{2/n}}{a(n+1)(n+2)}\right]^{3/2} t  \Bigg]}
\end{equation}
Let us consider several particular cases of Eq.(\ref{e1x4}). For $n=2$
Eq.(\ref{e1x4}) contains as particular case the Korteweg-deVries equation and 
the solution (\ref{c2s2}) is reduced to the famous ${\cosh}^2$ solitary wave
solution. For $n=1$ one of the the solutions of Eq.(\ref{e1x4}) is
\begin{equation}\label{ex_sol1}
F(x,t) = \frac{\gamma}{\cosh \Bigg[\left[ \frac{c\gamma^{2}}{6a}\right]^{1/2} x  - 
\frac{a}{b} \left[ \frac{c \gamma^{2}}{6 a}\right]^{3/2} t  \Bigg]}
\end{equation}
For $n=1/2$ Eq.(\ref{e1x4}) contains as particular case the modified
Korteweg-deVries equation and one of the solutions of (\ref{swx3}) is reduced to
\begin{equation}\label{ex_sol2}
F(x,t) = \frac{\gamma}{\cosh^{1/2} \Bigg[\left[ \frac{4 c\gamma^{4}}{15a}\right]^{1/2} x  - 
\frac{a}{4b} \left[ \frac{4 c \gamma^{4}}{15 a}\right]^{3/2} t  \Bigg]}
\end{equation}
\subsection{Case of a nonlinear partial differential equation that contains 
monomials only of even grade with respect to participating derivatives}
\par
Let us now consider the case of even derivatives for the equation
\begin{equation}\label{e2}
a F^p \frac{\partial ^\mu F}{\partial x^ \mu} + bF^q \frac{\partial^\nu 
F}{\partial t^\nu}  + cF^r  + d F^s =0
\end{equation}
Let us discuss the case  $\mu=4$, $\nu=2$. 
We need the following relationship
\begin{equation}\label{4der}
f_{\xi \xi \xi \xi} = n^4 f - 2n(n+1)(n^2+2n+2)f^{(n+2)/n} +
n(n+1)(n+2)(n+3)f^{(n+4)/n}
\end{equation}
We remember that we search for a solution of the kind $F(\xi)=\gamma f(\xi)$
where $f$ is solution of the simplest equation (\ref{se}). The application of
the balance procedure from the modified method of simplest equation leads to
the relationships $q=p+2/n$, $r=p+1$ together with one of the possibilities: $s=p+(n+2)/n$
or $s=p+(n+4)/n$. Let us consider the case $q=p+2/n$, $s=p+(n+2)/n$, $r=p+1$.
The equation (\ref{e2}) becomes
\begin{equation}\label{e2x1}
a \frac{\partial^4 F}{\partial x^4} + bF^{2/n} \frac{\partial^2 F}{\partial t^2} 
+cF + dF^{1 + 2/n} = 0
\end{equation}
Taking into account the forms of $F$ and $f$ we can
reduce Eq.(\ref{e2x1}) to the following system of nonlinear ordinary
differential equations
\begin{eqnarray}\label{e2x}
(b \beta^2 n^2 + d) \gamma^{2/n}-2 a \alpha^4  n (n+1) (n^2+2n+2) =0
\nonumber \\
-\beta^2 b \gamma^{2/n}+a\alpha^4(n+3)(n+2) =0
\nonumber \\
a\alpha^4 n^4 + c = 0
\end{eqnarray}
One possible solution of the system (\ref{e2x}) is
\begin{eqnarray}\label{solx}
\alpha&=& \frac{(-c a^3)^{1/4}}{(a n)} \nonumber \\
\beta&=& \Bigg[ 
\frac{(n^2+5n+6)d}{bn(n^3+n^2+2n+4)} \Bigg]^{1/2} 
\nonumber \\
\gamma&=& \bigg[ - \frac{c(n^3+n^2+2n+4)}{dn^3} \bigg]^{n/2}
\end{eqnarray}
and the corresponding solitary wave solution is
\begin{eqnarray}\label{sw1}
F(x,t) = \bigg[ - \frac{c(n^3+n^2+2n+4)}{dn^3} \bigg]^{n/2}{\Bigg/}\cosh^n 
\Bigg \{\frac{(-c a^3)^{1/4}}{(a n)} x + \nonumber \\
\Bigg[ 
\frac{(n^2+5n+6)d}{bn(n^3+n^2+2n+4)} \Bigg]^{1/2} t \Bigg \}
\end{eqnarray}
where $n$ can be arbitrary positive finite real number.
\par 
The second possibility is $q=p+2/n$ and $s=p+(n+4)/n$. 
In this case the equation (\ref{e2}) becomes
\begin{equation}\label{e2x2}
a \frac{\partial^4 F}{\partial x^4} + bF^{2/n} \frac{\partial^2 F}{\partial t^2} 
+ cF + dF^{1+4/n} = 0
\end{equation}
Taking into account that $F=\gamma f$ and the simplest equation for $f$ we can reduce
Eq.(\ref{e2x2}) to the system of
nonlinear algebraic equations 
\begin{eqnarray}\label{e2y}
-b n \beta^2 (n+1) \gamma^{2/n} + d \gamma^{4/n} + a \alpha^4  n (n+3) (n+2) (n+1) = 0 \nonumber \\
-\frac{1}{2} \beta^2 b \gamma^{2/n} n + a \alpha^4  (n+1) (n^2+2n+2) = 0
\nonumber \\
a \alpha^4 n^4 + c = 0 \nonumber \\
\end{eqnarray}
One possible solution of this system is
\begin{eqnarray}\label{soly}
\alpha &=& \frac{1}{n} \left( - \frac{c}{a} \right)^{1/4} \nonumber \\
\beta &=& \Bigg[- \frac{cd^{1/2}}{b} \frac{(n+1)(n^3+n^2+2n+4)+n(n^3+6n^2+11n+6)
}{n^3(n+1)[-c(n+1)(n^3+n^2+2n+4)]^{1/2}} \Bigg]^{1/2} \nonumber \\
\gamma &=& \Bigg[- \frac{c(n+1)(n^3+n^2+2n+4)}{dn^4} \Bigg]^{n/4}  
\end{eqnarray}
and the corresponding solitary wave solution is
\begin{eqnarray}\label{sw2}
F(x,t) &=&  \left[-\frac{c(n+1)(n^3+n^2+2n+4)}{dn^4} \right]^{n/4}{\Bigg/}
\cosh^n \Bigg\{ \frac{1}{n} \left( - \frac{c}{a} \right)^{1/4}x + \nonumber \\
&& \Bigg[- \frac{cd^{1/2}}{b} \frac{(n+1)(n^3+n^2+2n+4)+n(n^3+6n^2+11n+6)
}{n^3(n+1)[-c(n+1)(n^3+n^2+2n+4)]^{1/2}} \Bigg]^{1/2} t\Bigg\} \nonumber \\
\end{eqnarray}

\section{Concluding Remarks}
In this article we have continued our research from
\cite{v14} on the methodology connected to 
the method of simplest equation for obtaining exact solutions of nonlinear 
partial differential equations. We have formulated a theorem about a 
solitary wave solution of kind $1/\cosh^n(\xi)$, $\xi=\alpha x + \beta t$ 
that may help us to find
solitary wave solutions of a large class of nonlinear partial differential 
equations. The developed methodology is applied to two classes of nonlinear 
PDEs. We note that the methodology can be applied also to more complicated 
nonlinear PDEs, e.g., to  equations containing together monomials of odd and
even grades with respect to participating derivatives. Let us note that the limits of the 
applicability of the methodology will be reached when the size of the 
system of nonlinear algebraic equations becomes large. Then the number of 
parameters participating in the solved equation and in the solution can become
smaller than the number of equations. Another problem may arise if some of the
algebraic equations have nonlinearity of high order and because of this an
analytical solution is impossible to be obtained. But as we have shown the methodology
is effective and leads to exact solitary wave solutions to many nonlinear partial
differential equations.

\end{document}